\documentclass[12pt]{article} %
\usepackage{amsmath,amssymb,amsthm}
\numberwithin{equation}{section} %
\theoremstyle{plain}
   \newtheorem{thm}{\hspace{\parindent}{\sc Theorem}}[section] %
   \newtheorem{pro}[thm]{\hspace{\parindent}Proposition}
   \newtheorem{cor}[thm]{\hspace{\parindent}Corollary}
   \newtheorem{lem}[thm]{\hspace{\parindent}Lemma}
\theoremstyle{remark} %
   \newtheorem{rem}{\hspace{\parindent}Remark}[section] %
\newtheorem{exmp}{\hspace{\parindent}Example}[section]%
\newcommand{\alphahat}{\widehat{\alpha}}
\newcommand{\alphahatj}{\widehat{\alpha}^{(j)}}
\newcommand{\betahat}{\widehat{\beta}}
\newcommand{\bR}{\mathbb{R}}

\newcommand{\bq}{\bold q}

\newcommand{\Czerospace}{C^{\infty}_0(\mathbb{R}^d)}

\newcommand{\domain}{\It \times \bR^d}
\newcommand{\Gepsilon}{G_{\epsilon}}
\newcommand{\It}{I_{T}}
\newcommand{\kddelta}{K_{D\Delta}}
\newcommand{\lambdamax}{\lambda_{\max}}

\newcommand{\qts}{q^{t,s}_{x,y}}
\newcommand{\qdelta}{q_{\Delta}}

\newcommand{\Sspace}{{\cal S}}

\newcommand{\supp}{\text{supp}}

\newcommand{\xidelta}{\xi_{\Delta}}
\def\dbar{{\mathchar'26\mkern-12mud}}
\begin{document}
\title{Notes on the Feynman path integral for the Dirac equation}
\author{Wataru Ichinose\thanks{Research partially supported by Grant-in-Aid for Scientific Research 
No.26400161,  Ministry of Education, Culture, Sports, Science and 
Technology, Japanese
Government.}} %
\date{}
\maketitle %
\begin{quote}
{\small Department of Mathematics, Shinshu University,
Matsumoto 390-8621, Japan. \\ E-mail: ichinose@math.shinshu-u.ac.jp}%
\end{quote}\par
\noindent{\small 
	This paper is a continuation of  the author's preceding one.  In the preceding paper the author has rigorously constructed the Feynman path integral  for the Dirac equation  in the form of the sum-over-histories, satisfying  the superposition principle, over all paths of one electron in  space-time that goes in any direction  at any speed, forward and backward in time with a finite number of turns.  In the present paper, first we will generalize the results in the preceding paper and secondly prove in a direct way that  our  Feynman path integral satisfies the unitarity principle and the causality one. 	}
\vspace{0.3cm}
\section{Introduction}
This paper is a continuation of  the author's preceding one \cite{Ichinose 2014}.   Let   $T > 0$ be an arbitrary constant. We will study  the Dirac equation
\begin{align} \label{1.1}
& i\hbar\frac{\partial u}{\partial t} (t)  = H(t)u(t) \notag \\
& :=\left[c\sum_{j=1}^d \alphahat^{(j)} \left(\frac{\hbar}{i}\frac{\partial}{\partial x_j} - eA_j(t,x) \right) + \betahat mc^2 + e V(t,x)I_N\right]u(t),
\end{align}
where $t \in I_T :=[-T, T]$, $x = (x_1,\dotsc,x_d) \in \mathbb{R}^d$, $u(t) = {}^t(u_1(t),\dotsc,u_N(t)) 
\in \mathbb{C}^N$, $(V(t,x),A(t,x)) = (V,A_1,\dotsc,A_d) \in \bR^{d+1}$ is an electromagnetic potential, $\alphahat^{(j)} (j = 1,2,\dotsc,d)$ and $\betahat$ are constant $N\times N$ Hermitian matrices, $I_N$ is the $N\times N$ identity matrix, $c$ is the velocity of light, $\hbar$ is the Planck constant and $e$ is the charge of an electron.  Though the relations
\begin{equation} \label{1.2}
\alphahat^{(j)}\alphahat^{(k)} + \alphahat^{(k)}\alphahat^{(j)} = 2\delta^{jk}I_N, \  \alphahat^{(0)} = \betahat 
\end{equation}
for $j, k = 0,1,\dotsc,d$ are assumed for the genuine Dirac equation (cf. (8) of \S 67 in \cite{Dirac}), in the present paper $\alphahat^{(j)}$ and $\betahat$ are assumed to be only Hermitian as in \cite{Ichinose 2014}, where $\delta^{jk}$ denotes the Kronecker delta.
 For the sake of simplicity we suppose $\hbar = 1$ and $e = 1$ hereafter, and will sometimes omit $I_N$.
 \par
In the preceding paper \cite{Ichinose 2014} the author has rigorously constructed   the Feynman path integral for the Dirac equation \eqref{1.1} in the form of the sum-over-histories,
satisfying the superposition principle, over all possible paths of one electron in   space-time  that goes in any direction at any speed, forward and backward in time with a finite number of turns.  In addition, the author has proved that the Feynman path integral constructed above satisfies the Dirac equation \eqref{1.1}. 
	It should be noted that Feynman had said for the application of his path integral to quantum electrodynamics that the electron goes in any direction at any speed  forward and backward in time, as seen  on p.376 of \cite{Dyson 1980}, in \cite{Feynman positron} and on p.388 of \cite{Schweber}.  
\par
	In the present paper, first we will generalize the results  in \cite{Ichinose 2014} and secondly prove in a direct way that our Feynman path integral satisfies the unitarity principle and the causality one.  We basically owe  our arguments in their proofs to the theory of pseudo-differential operators.
		\par
	First,  we will prove that the assumptions about a magnetic strength tensor can be generalized for the Feynman path integral to be determined.  The assumptions about this haven't been able to be generalized for a long time since \cite{Ichinose 1999} in 1999.  Our proof will be obtained by returning to the original idea of Theorem 3.7 in \cite{Ichinose 1997}.
	\par
	The second generalization is  in the $L^2$ space.  In the present paper we will determine  the Feynman path integral in the form of the sum-over-histories, satisfying the superposition principle, over all possible paths of one electron   that goes in any direction at any speed,  forward and backward in time  particularly with a countably infinite number of turns.  Here, $L^2 = L^2(\bR^d)$ denotes the space of all square integrable functions on $\bR^d$ with  inner
product $(f,g) := \int f(x)\overline{g(x)}dx$ and  norm $\Vert f\Vert$, where $\overline{g(x)}$ is the complex conjugate of $g(x)$. Our proof will be obtained  as in the proof of Theorem 2.1 in \cite{Ichinose 2014} by using the  estimate \eqref{3.8} in the present paper.
\par
  Next, we will study the properties of our Feynman path integral for the Dirac equation. First, we will prove that the Feynman path integral makes a unitary operator on the $(L^2)^N$ space. This result gives another proof of the unitarity on $(L^2)^N$ of  the fundamental solution to the Dirac equation \eqref{1.1}, which  is well known in the theory of partial differential equations. Our proof is based on the estimate \eqref{3.8} in the present paper too.
	\par
     Secondly, we will prove that our Feynman path integral 
satisfies the causality principle, i.e.  has the speed  not exceeding the velocity of light of propagation of disturbances.  This result gives another proof that every  solution to the Dirac equation has the same property, which  is also well known in the theory of partial differential equations.  Our proof is based on the Paley-Wiener theorem (cf.Theorem IX.11 in \cite{Reed-Simon}), which theorem characterizes the size of the support of  functions by their Fourier transforms.
As seen above, to construct the Feynman path integral we  use  the paths,  of one electron in  space-time, violating  causality.  Consequently our result, that the Feynman path integral satisfies causality, implies that the probability amplitudes for such paths are completely canceled out by the effect of interference among themselves and other probability ones, as argued in \S 1-3 of \cite{Feynman-Hibbs}. \par
 Our proof that the Feynman path integral satisfies unitarity and causality is more direct than the proof in the theory of partial differential equations that every solution to the Dirac equation has the same properties. Our results are yielded  from (4.1) and (4.6), and \eqref{2.7} and the Paley-Wiener theorem, respectively.
\par
  The plan of the present paper is as follows.  In \S 2 we will state the results on the Feynman path integral. In \S 3 we will   prove them.  In \S 4 we will state the results on unitarity and causality for  the Feynman path integral and  prove them.
\section{Results on the Feynman path integral}
For an $x = (x_1,\dotsc,x_d) \in R^d$ and a multi-index
$\alpha = (\alpha_1,\dotsc,\alpha_d)$,  we  write $|x| = \sqrt{\sum_{j=1}^d x^2_j}$, $|\alpha| = 
\sum_{j=1}^d
\alpha_j$, $x^{\alpha} =  x_1^{\alpha_1}
\cdots  x_d^{\alpha_d}, \partial_{x_j} = \partial /\partial
x_j$ and  $\partial_x^{\alpha} = \partial_{x_1}^{\alpha_1}
\cdots \partial_{x_d}^{\alpha_d}$.  In the present paper we often use symbols $C, C_{\alpha}, C_{\alpha,\beta}$ and $C_a$ to write down constants, although these values are different in general. \par
   Let us write  the classical Hamiltonian function
\begin{equation} \label{2.1}
  {\cal H}(t,x,p) =  c\sum_{j=1}^d \alphahatj \bigl(p_j - A_j(t,x)\bigr) + \betahat mc^2 +  V(t,x) I_N
\end{equation}
for $H(t)$ defined by \eqref{1.1} as in (23) on p.261 of \cite{Dirac}, where $p \in \bR^d$ is the canonical momentum.  We write the kinetic momentum as $\xi := p - A(t,x) \in \bR^d$.  Then the classical Lagrangian function is given by
\begin{align} \label{2.2}
  & {\cal L}(t,x,\dot{x},\xi) = p\cdot \dot{x} - {\cal H}(t,x,p) \notag \\
  & = \xi\cdot \dot{x} + \dot{x}\cdot A(t,x) - V(t,x)I_N - 
  (c\alphahat\cdot\xi + \betahat mc^2),
\end{align}
where $p\cdot \dot{x} = \sum_{j=1}^d p_j\dot{x}_j, \alphahat = (\alphahat^{(1)},\dotsc,\alphahat^{(d)})$ and $\alphahat\cdot\xi = \sum_{j=1}^d \alphahatj\xi_j$. \par
  Let $t$ and $s$ be in $\It$ such that $t \not= s$.  For $x$ and $y$  in $\bR^d$ we define
\begin{equation}  \label{2.3}
     q^{t,s}_{x,y}(\theta) := y + \frac{\theta - s}{t-s}(x - y)
\end{equation}
in $s \leq \theta \leq t$ or $t \leq \theta \leq s$. Let $\xi \in R^d$ and consider a path $(\qts(\theta),\xi) \in \bR^{2d}$ in phase space.  The classical action for this path is given by 
\begin{align}  \label{2.4}
 &  S(t,s;x,\xi,y) := \int_s^t{\cal L}(\theta,\qts(\theta),\dot{q}^{t,s}_{x,y}(\theta), \xi)d\theta = (x - y)\cdot\xi\notag \\ 
       &   + \int_s^t\Bigl\{\dot{q}^{t,s}_{x,y}(\theta)\cdot A(\theta,q^{t,s}_{x,y}(\theta)) - V(\theta,q^{t,s}_{x,y}(\theta))\Bigr\}d\theta -(t - s)(c\alphahat\cdot\xi + \betahat mc^2) \notag \\
       & = (x - y)\cdot\xi +    (x - y)\cdot\int_0^1
 A(t -\theta\rho,x - \theta(x - y))d\theta
       &  \notag\\
 &  -\rho\int_0^1V(t -\theta\rho,x - \theta(x - y))d\theta- \rho(c\alphahat\cdot\xi + \betahat mc^2), \quad \rho = t - s
       \end{align} 
from \eqref{2.2}, where $\dot{q}^{t,s}_{x,y}(\theta) = d\qts(\theta)/d\theta$.  The matrices $\alphahat^{(j)}$ and $\betahat$ are assumed to be Hermitian and so is  $S(t,s;x,\xi,y)$.
 Noting \eqref{2.4}, we will define 
       \begin{equation} \label{2.5}
       S(s,s;x,\xi,y) := (x - y)\cdot\xi +    (x - y)\cdot\int_0^1
 A(s,x - \theta(x - y))d\theta,
       \end{equation}
       which we write $\displaystyle \int_s^s{\cal L}(\theta,q^{s,s}_{x,y}(\theta),\dot{q}^{s,s}_{x,y}(\theta), \xi)d\theta$ formally.　\par
       Let $t_i \in \It$ and $t_f \in \It$ be an initial time and a final one respectively, where $t_i \leq t_f$ or $t_i > t_f$.  Take $\tau_j \in \It\ (j = 1,2,\dotsc,\nu - 1)$ and consider a time-division $\Delta := \{\tau_j\}_{j=1}^{\nu-1}$, where
$\tau_j \leq \tau_{j+1}$ or $\tau_j > \tau_{j+1}$.  We set $\tau_0 = t_i$ and $\tau_{\nu} = t_f$.  We take a point $x \in \mathbb{R}^d$ and fix it.
 Taking points
$x^{(j)} \in \mathbb{R}^d\ (j = 0,1,\dotsc,\nu-1)$ arbitrarily,  we define a piecewise linear  path $(\Theta_{\Delta},\qdelta(x^{(0)},\dotsc,x^{(\nu-1)},x))$ in  space-time $\domain$ by joining  $(\tau_j,x^{(j)})\ (j = 0,1, \dotsc,\nu, x^{(\nu)} = x)$  in order. 
Next,  taking  points $\xi^{(j)} \in \mathbb{R}^d\ (j = 0,1,\dotsc,\nu-1)$ arbitrarily,  we also define a piecewise constant path $(\Theta_{\Delta},\xidelta(\xi^{(0)},\dotsc,\xi^{(\nu-1)}))$ in $\domain$ by using $\xidelta$ that has the value $\xi^{(j)}\ (j = 0,1, \dotsc,\nu-1)$ for $\theta \in [\tau_j,\tau_{j+1}]$ if $\tau_j \leq \tau_{j+1}$ or $\theta \in [\tau_{j+1},\tau_{j}]$ if $\tau_{j+1} < \tau_{j}$.
 We note that the paths $(\Theta_{\Delta},\qdelta)$ and $(\Theta_{\Delta},\xidelta)$  go in any direction forward and backward in time and that $\qdelta$ has any speed, even the infinite speed. \par
    Let us consider the path $(\Theta_{\Delta},\qdelta(x^{(0)},\dotsc,x^{(\nu-1)},x), \xidelta(\xi^{(0)},\dotsc,\xi^{(\nu-1)})) $ in $\It \times \mathbb{R}^{2d}$ connecting $(t_i,x^{(0)},\xi^{(0)})$ with $(t_f,x,\xi^{(\nu-1)})$.
We define the probability amplitude $\exp *iS(t_f,t_i,\qdelta,\xidelta)$ for this path in terms of the classical action \eqref{2.4} and \eqref{2.5} by the product of  unitary matrices
\begin{align} \label{2.6}
 & \exp i\int_{\tau_{\nu-1}}^{t_f} {\cal L}
 (\theta,q^{t_f,\tau_{\nu-1}}_{x,x^{(\nu-1)}}(\theta),\dot{q}^{t_f,\tau_{\nu-1}}_{x,x^{(\nu-1)}}(\theta),\xi^{(\nu-1)})d\theta
  \cdot
  \exp i\int_{\tau_{\nu-2}}^{\tau_{\nu-1}}
  {\cal L}
 (\theta,q^{\tau_{\nu-1},\tau_{\nu-2}}_{x^{(\nu-1)},x^{(\nu-2)}}(\theta),
 \notag  \\
  & \dot{q}^{\tau_{\nu-1},\tau_{\nu-2}}_{x^{(\nu-1)},x^{(\nu-2)}}(\theta),\xi^{(\nu-2)})d\theta \cdots \cdot \exp i\int_{t_i}^{\tau_{1}}
  {\cal L}
 (\theta,q^{\tau_{1},t_i}_{x^{(1)},x^{(0)}}(\theta),\dot{q}^{\tau_{1},t_i}_{x^{(1)},x^{(0)}}(\theta),\xi^{(0)})d\theta. 
\end{align} \par
   Let $\Sspace = \Sspace(\mathbb{R}^d)$ be  the Schwartz  space of all rapidly decreasing functions on $\bR^d$ with the well-known topology.  We take  a  function $\chi \in \Sspace(\mathbb{R}^d)$ such that $\chi(0) = 1$.  Let $f = {}^t(f_1,\dotsc,f_N) \in \Sspace(\mathbb{R}^d)^N$ and  define an approximation $K_{D\Delta}(t_f,t_i)f$ of the Feynman path integral  for the Dirac equation \eqref{1.1} by
\begin{align} \label{2.7}
   &  K_{D\Delta}(t_f,t_i)f =  \iint e^{*iS(t_f,t_i,\qdelta,\xidelta)} 
f(\qdelta(t_i)) {\cal D}\qdelta{\cal D}\xidelta 
 \notag \\
& :=  \lim_{\epsilon \rightarrow +0} \int
\dotsi\int e^{*iS(t_f,t_i,\qdelta,\xidelta)} f(x^{(0)})\prod_{j=0}^{\nu-1}\left\{\chi(\epsilon x^{(j)})\chi(\epsilon \xi^{(j)})\right\} dx^{(0)}\cdots dx^{(\nu-1)} \notag \\ 
&  \ \ \cdot\dbar \xi^{(0)}\cdots \dbar \xi^{(\nu-1)},
   \end{align}
where $\dbar \xi^{(j)} = (2\pi)^{-d}d\xi^{(j)}$.  As stated in Theorem 2.A below,  $K_{D\Delta}(t_f,t_i)f$ is determined independently of the choice of $\chi$.  Hence the integral \eqref{2.7} is often called the oscillatory integral and  written as 
\begin{equation*}
\text{Os}-\int\dotsi\int e^{*iS(t_f,t_i,\qdelta,\xidelta)} f(x^{(0)}) dx^{(0)}\cdots dx^{(\nu-1)}\dbar \xi^{(0)}\cdots \dbar \xi^{(\nu-1)}
\end{equation*}
(cf. p. 45 of \cite{Kumano-go}). \par
   For $f = {}^t(f_1,\dotsc,f_N) \in L^2(\mathbb{R}^d)^N$ we write its norm $\sqrt{\sum_{j=1}^N\Vert f_j\Vert^2}$ as $\Vert f\Vert$. Let $E(t,x) = (E_1,\dots, E_d) \in \mathbb{R}^d$ and  $\bigl(B_{jk}(t,x)\bigr)_{1\leq j < k \leq d} \in \bR^{d(d-1)}$ be electric strength and a magnetic strength tensor, respectively.     In Theorems 2.1 and 2.2 of \cite{Ichinose 2014} we have proved the following. \par
\vspace{0.3cm}
\noindent{\bf Theorem 2.A}.  {\it Let $\partial_x^{\alpha}E_j(t,x)\ (j = 1,2,\dotsc,d), \partial_x^{\alpha}B_{jk}(t,x)\ (1 \leq j < k \leq d)$ and $\partial_tB_{jk}(t,x)$ be continuous in $I_T\times R^{d}$ for all $\alpha$.  We assume
\begin{equation} \label{2.8} 
 |\partial_x^{\alpha}E_j(t,x)| \leq C_{\alpha},
     \quad |\alpha| \geq 1,
     \end{equation}
     \begin{equation} \label{2.9}
 |\partial_x^{\alpha}B_{jk}(t,x)| \leq C_{\alpha}<x>^{-(1 + \delta_{\alpha})},
      \quad  |\alpha| \geq 1
\end{equation}
in $\It\times R^{d}$ with constants $\delta_{\alpha} > 0$ for $j = 1,2,\dotsc,d$ and $1\leq j < k \leq d$, where $<x> = \sqrt{1 + |x|^2}$.  Let $(V,A)$ be an electromagnetic potential inducing $E(t,x)$ and 
$(B_{jk}(t,x))_{1\leq j < k \leq d}$ via equations
\begin{align} \label{2.10}
     & E = -\frac{\partial A}{\partial t} - \frac{\partial V}{\partial 
x}, \notag\\
          & B_{jk} =  \frac{\partial A_k}{\partial x_j}  -\frac{\partial 
A_j}{\partial x_k}
\quad (1 \leq j <  k \leq d)
\end{align}
such that $V, \partial_{x_j}V, \partial_{t}A_k$ and $\partial_{x_j}A_k\ (j,k = 1,2,\dotsc,d)$ are continuous in $\It\times R^{d}$, where  $\partial V/\partial x = (\partial V/\partial x_1,\dots,\partial V/\partial x_d)$. 
We take $t_i$ and $t_f $ in $\It$.  Let $\tau_j \in \It\ (j = 1,2,\dots,\nu-1)$, determine $\Delta = \{\tau_j\}_{j=1}^{\nu-1}$ and define $K_{D\Delta}(t_f,t_i)f$ for $f \in \Sspace^N$ by \eqref{2.7}.  \par
Then we have: (1) $K_{D\Delta}(t_f,t_i)$ on $\Sspace^N$ is determined independently of the choice of $\chi \in \Sspace$ and can be extended to a bounded operator on $(L^2)^N$.  (2)  Let $f \in (L^2)^N$. Let $L_0 \geq 0$ be an arbitrary constant and consider only time-divisions $\Delta$ satisfying
    \begin{equation} \label{2.11}
    \sum_{j=0}^{\nu-1}|\tau_{j+1} - \tau_j| \leq L_0.
    \end{equation}
Then, as   $|\Delta| := \max_{j=0,1,\dots,\nu-1}|\tau_{j+1}-\tau_j|  \rightarrow 0$,
    $K_{D\Delta}(t_f,t_i)f$ converges  in $(L^2)^N$  uniformly with respect to $t_f$and $t_f$ in $\It$,  and this limit $K_{D}(t_f,t_i)f$, called the Feynman path integral,  is determined independently of the choice of  $L_0$.
     (3) $K_{D}(t_f,t_i)f$ for  $f \in (L^2)^N$ belongs to ${\cal E}^0_{t_f}(\It;(L^2)^N)$ and satisfies the Dirac equation \eqref{1.1} in  the distribution sense with $u(t_i) = f$, where ${\cal E}^j_{t_f}(\It;(L^2)^N)\ (j = 0,1,\dotsc)$ denotes the space of all $(L^2)^N$-valued j-times continuously differentiable functions on $\It$.  (4)  Let $\psi(t,x)$ be a real-valued function such that $\partial_{x_j}\partial_{x_k}\psi(t,x)$ and $\partial_{t}\partial_{x_j}\psi(t,x)$ $(j,k = 1,2,\dotsc,d)$ are continuous in $\It \times \mathbb{R}^d$.  We  consider the gauge transformation 
\begin{equation}  \label{2.12}
    V' = V -\frac{\partial\psi}{\partial  t}, \quad  A'_j = A_j + \frac{\partial\psi}{\partial  x_j}\quad (j = 1,2,\dots,d)
\end{equation}
and write \eqref{2.7} for this $(V',A')$ as $K'_{D\Delta}(t_f,t_i)f$.  Then we have a formula 
\begin{equation}  \label{2.13}
     K'_{D\Delta}(t_f,t_i)f  = e^{i\psi (t_f,\cdot)}K_{D\Delta}(t_f,t_i)\left(e^{-i\psi (t_i,\cdot)}f\right)
\end{equation}
for all $f \in (L^2)^N$ and so have the same formula for $K_{D}(t_f,t_i)f$.
}
\par
Let $M$ and $a$ be  positive integers.  We introduce  the weighted Sobolev spaces
 $B^a_M(\mathbb{R}^d)^N  := \{f \in L^2(\mathbb{R}^d)^N;\
\|f\|_{B^a_M} := \|f\| + \sum_{|\alpha| = aM} \|x^{\alpha}f\| +
\sum_{|\alpha| = a}
\|\partial_x^{\alpha}f\| 
\\
< \infty\}$.  Let
$B^{-a}_M(\mathbb{R}^d)^N$ denote their dual spaces and set 
$B^0_M(\mathbb{R}^d)^N: = L^2(\mathbb{R}^d)^N$.  \par
\noindent{\bf Theorem 2.B}.  {\it  Besides the assumptions of Theorem 2.A we assume the following: (1)  We have
\begin{equation} \label{2.14}
|\partial_x^{\alpha}A_j(t,x)| \leq C_{\alpha},
     \quad |\alpha| \geq 1
\end{equation}
in $I_T \times \bR^d$ for $j = 1,2,\dotsc,d$.  (2)   There exists an integer $M \geq  1$ such that
\begin{equation} \label{2.15}
|\partial_x^{\alpha}\partial_tA_j(t,x)| \leq C_{\alpha}<x>^M
\end{equation}
for all $\alpha$ in $\domain$.  Then we have: (1) $K_{D\Delta}(t_f,t_i)$ on $\Sspace^N$ can be extended to a bounded operator on $(B^a_M)^N\ (a = 0,1,\dotsc)$.  (2) Let $f \in (B^a_{M+1})^N$ and  $L_0 \geq 0$  an arbitrary constant.  Then, as $|\Delta| \rightarrow 0$ under the assumption \eqref{2.11}, $K_{D\Delta}(t_f,t_i)f$ converges to  $K_{D}(t_f,t_i)f$ in $(B^a_{M+1})^N$ uniformly with respect to  $t_f$ and $t_i$ in $\It$.
}
\begin{rem}
In \cite{Ichinose 2014} we used   $\chi \in C_0^{\infty}(\bR^d)$, i.e. an infinitely differentiable function in $\bR^d$ with compact support, to define $\kddelta(t_f,t_i)f$ by \eqref{2.7} in place of $\chi \in \Sspace(\bR^d)$.  However, the proof of Proposition 3.2 in \cite{Ichinose 2014} or Proposition  3.2 in the present paper assures us that Theorems 2.A  and 2.B above remain true.
\end{rem}
\begin{rem}
In Theorem 2.2 in \cite{Ichinose 2014} we assumed 
\begin{equation} \label{2.16}
|\partial_x^{\alpha}V(t,x)| \leq C_{\alpha}<x>^M,
     \quad |\alpha| \geq 1
 \end{equation}
besides \eqref{2.14} and \eqref{2.15}.  We note that \eqref{2.16} are derived from \eqref{2.8}, \eqref{2.10} and \eqref{2.15}.
\end{rem}
    We will state  Theorems 2.1 and 2.2 as the main results on the Feynman path integral.
   \begin{thm} 
   {\it 
   In Theorems 2.A and 2.B we replace the assumption \eqref{2.9} with \eqref{2.14}, \eqref{2.15} and
     \begin{equation} \label{2.17}
 |\partial_x^{\alpha}\partial_tB_{jk}(t,x)| \leq C_{\alpha}<x>^{-(1 + \delta_{\alpha})},
      \quad  |\alpha| \geq 1
\end{equation}
for $1 \leq j < k \leq d$ in $\domain$ where  $\delta_{\alpha} > 0$ are constants.  Then the same assertions as in Theorems 2.A and 2.B hold respectively.
   }
   \end{thm}
   \begin{rem}
   If $A_j\ (j= 1,2,\dots,d)$ satisfying \eqref{2.14} are independent of $t \in \It$, \eqref{2.15} holds and \eqref{2.17} follows from \eqref{2.10}.  Hence we can see that Theorem 2.1 gives new results.  The assumption \eqref{2.9} hasn't been able to be generalized for a long time since \cite{Ichinose 1999} in 1999.
   \end{rem}
   We will consider the Feynman path integral in the $L^2$ space.
\begin{thm} 
We suppose the  assumptions of Theorem 2.A or make  in Theorem 2.A the same replacement of the assumption \eqref{2.9} as in Theorem 2.1.
   For time-divisions $\Delta = \{\tau_j\}_{j=1}^{\nu -1}$ we set
\begin{equation} \label{2.18}
 \sigma(\Delta) := \sum_{j=0}^{\nu-1}(\tau_{j+1}-\tau_j)^2.
     \end{equation} 
Then we obtain: (1)  Under the assumption $\sigma(\Delta) \leq 1$ we have
\begin{equation}  \label{2.19}
    \Vert K_{D\Delta}(t_f,t_i)f \Vert \leq e^{K_0\sigma(\Delta)}\Vert f \Vert
    \end{equation}
for all $t_i, t_f$ in $\It$ with a constant $K_0 \geq 0$.  (2)  Let $f \in (L^2)^N$.  
  Then, as $\sigma(\Delta) \rightarrow 0$,
$K_{D\Delta}(t_f,t_i)f$ converges to the Feynman path integral $K_{D}(t_f,t_i)f$  in $(L^2)^N$ 
    uniformly with respect to $t_f $ and $t_i $ in $\It$.

\end{thm}
  %
  %
  \begin{rem}
  Theorem 2.2 gives a generalization of Theorem 2.A and a part of Theorem 2.1 because of 
\begin{equation*} 
 \sigma(\Delta) = \sum_{j=0}^{\nu-1}(\tau_{j+1}-\tau_j)^2 \leq |\Delta| \sum_{j=0}^{\nu-1}|\tau_{j+1}-\tau_j|.
     \end{equation*} 
  \end{rem}
 The corollary below follows from (2)  of Theorem 2.2.
 \begin{cor}
 Consider time-divisions $\Delta(n) := \{\tau_j^{(n)}\}_{j=1}^{\nu(n)-1}\ (n =1,2, \dots)$ such that $\lim_{n\rightarrow \infty}\sigma(\Delta(n)) = 0$  and  for each $n$ there exist $j_k$ and $j'_k\ (k = 1,2,\dots,n)$ satisfying  $\tau_{j_k}^{(n)} = T$ and $\tau_{j'_k}^{(n)}= -T$.   Let    $f \in (L^2)^N$.  Then, under the assumptions of Theorem 2.2 we have
\begin{equation}  \label{2.20}
     \lim_{n\rightarrow \infty}K_{D\Delta(n)}(t_f,t_i)f  = K_{D}(t_f,t_i)f
     \end{equation}
in  $ (L^2)^N$  uniformly with respect to $t_f$ and $t_i$ in $\It$.
 \end{cor}
 \begin{exmp}
 We can easily construct time-divisions $\Delta(n)\ (n = 1,2,\dots)$ satisfying the properties stated in Corollary 2.3.
In fact, let $t_i < t_f$ and take $\nu(n) = (2n+1)n^2$.  We can easily determine 
$\Delta(n) =  
\{\tau_j \}_{j=1}^{\nu(n)-1}$ such that
\begin{align*}  
  &  t_i < \tau_1 <\tau_2 <\dotsc < \tau_{j_1}=T >  \tau_{j_1+1} > \dotsc > \tau_{j'_1}=-T < \tau_{j'_1+1} < \\
    &  \dotsc < \tau_{j_2} =T >  \tau_{j_2+1} > \dotsc > \tau_{j'_n} =-T< \tau_{j'_n+1}  <\dotsc < \tau_{\nu(n)-1} < t_f
     \end{align*}
 and $|\tau_{j+1} - \tau_j| \leq 2T/n^2$.  For example, we have only to take $j_k =(2k-1)n^2$ and $j'_k = 2kn^2$ for $k = 1,2,\dots,n$. Then we have 
\begin{equation*} 
 \sigma(\Delta(n)) = \sum_{j=0}^{\nu(n)-1}(\tau_{j+1}-\tau_j)^2 \leq (2n+1)n^2\left(\frac{2T}{n^2}\right)^2,
     \end{equation*} 
which tends to zero as $n \rightarrow \infty$.
 \end{exmp}
 \begin{rem}
 The left-hand side of \eqref{2.20} gives the Feynman path integral in the form of sum-over-histories over all paths of one electron that goes forward and backward in time with a countably infinite number of turns.
 \end{rem}
 
 %
 %
 %
 \section{Proofs of Theorems 2.1 and 2.2}
 Let $t$ and $s$ be in $\It$.
 We set
\begin{align} \label{3.1}
&  \Psi_j(t,s;x,y,z) := - \int_0^1A_j(s,z+ \theta(x-z))d\theta  \notag \\
& + (t-s)\int_0^1\int_0^1\sigma_1 E_j(t-\sigma_1(t-s),y+\sigma_1(z-y)+\sigma_1\sigma_2(x-z))d\sigma_1d\sigma_2 \notag \\
& + \sum_{k=1}^d(y_k-z_k)\int_0^1\int_0^1\sigma_1B_{jk}(t-\sigma_1(t-s),y+\sigma_1(z-y)+\sigma_1\sigma_2(x-z))d\sigma_1d\sigma_2
\end{align}
as in (3.7) of  \cite{Ichinose 2014} and 
\begin{align} \label{3.1.1}
&  \Psi'_j(t,s;x,y,z) := - \int_0^1A_j(s,z+ \theta(x-z))d\theta  \notag \\
& + (t-s)\int_0^1\int_0^1\sigma_1 E_j(t-\sigma_1(t-s),y+\sigma_1(z-y)+\sigma_1\sigma_2(x-z))d\sigma_1d\sigma_2 \notag \\
&  +  (t-s)\int_0^1 d\theta \sum_{k=1}^d(y_k-z_k)\int_0^1\int_0^1\sigma_1(1-\sigma_1)\frac{\partial B_{jk}}{\partial t}(s + \theta(t-s)(1-\sigma_1), \notag\\
&  y+\sigma_1(z-y)+\sigma_1\sigma_2(x-z))d\sigma_1d\sigma_2
\end{align}
for $j = 1,2,\dots,d$.

 \begin{lem}
We have
\begin{equation} \label{3.1.2}
 (x - z)\cdot \Psi(t,s;x,y,z) =  (x - z)\cdot \Psi'(t,s;x,y,z).
\end{equation}
Under the assumptions  \eqref{2.8}, \eqref{2.14} and \eqref{2.17} we have
\begin{equation} \label{3.2}
            |\partial_x^{\alpha}\partial_y^{\beta}\partial_{z}^{\gamma}
            \Psi'_j(t,s;x,y,z)|  \leq C_{\alpha,\beta,\gamma},
\ |\alpha + \beta  + \gamma| \geq 1
\end{equation}
in $\It^2\times \bR^{3d}$ for $j = 1,2,\dots,d$.
 \end{lem}
 \begin{proof}
   Let us return to the proof of Lemma 3.4 in \cite{Ichinose 2014}.  
 Let $\Lambda$ be the 2-dimensional plane in $\domain$ determined in (3.8) of  \cite{Ichinose 2014}. Then we have
\begin{equation*} 
 \lim_{t\rightarrow s\pm 0}\iint_{\Lambda}d(A\cdot dx - Vdt) = 0.
\end{equation*}
 Hence from the proof of Lemma 3.2 in \cite{Ichinose 1997} we can see
\begin{align*} 
& \sum_{j=1}^d(x_j-z_j)\sum_{k=1}^d(y_k-z_k)\int_0^1\int_0^1\sigma_1B_{jk}(s,y  +\sigma_1(z-y)\\
&\qquad  +\sigma_1\sigma_2(x-z))d\sigma_1d\sigma_2 = 0
\end{align*}
for all $(x,y,z) \in \bR^{3d}$.  Consequently, subtracting the coefficient of $x_j-z_j$ in the above from $\Psi_j(t,s;x,y,z)$, we get \eqref{3.1.1} and \eqref{3.1.2}.
\par
It follows from \eqref{2.8} and \eqref{2.14} that the first term and the second one on the right-hand side of \eqref{3.1.1} satisfy \eqref{3.2}.
Applying Lemma 3.5 in \cite{Ichinose 1997} to the third term on the right-hand side of \eqref{3.1.1}, we can see from \eqref{2.17} that the third term satisfies \eqref{3.2} as well.  Thus, the proof is complete.
 \end{proof}
 Now we will prove Theorem 2.1.  Let us define an operator  on $\Sspace^N$ by
\begin{equation}  \label{3.5}
     (G_{\epsilon}(t,s)f)(x) = \iint e^{iS(t,s;x,\xi,y)}f(y)
     \chi(\epsilon \xi)dy\dbar\xi
\end{equation}
for $\epsilon > 0$ in terms of \eqref{2.4} and \eqref{2.5} as in (1.12) of \cite{Ichinose 2014}, where $\chi \in \Sspace(\bR^d)$ such that $\chi(0) = 1$.   Let $G_{\epsilon}(t,s)^*$ denote the formally adjoint operator of $G_{\epsilon}(t,s)$.
We will do use Lemma 3.1.  Then, noting Lemma 3.4 in \cite{Ichinose 2014}, as in the proof of Proposition 3.5 in \cite{Ichinose 2014} we can prove
\begin{align}  \label{3.5.2} 
     & \left(\Gepsilon(t,s)^*\Gepsilon(t,s)f\right) (x) = \iint e^{i(x-z)\cdot\xi} dz\dbar\xi \iint e^{-i\eta\cdot w}  e^{i(t-s)(c\alphahat\cdot \xi + c\alphahat\cdot\Psi' + \betahat mc^2)}      \notag \\
       &  \times  e^{-i(t-s)(c\alphahat\cdot \xi + c\alphahat\cdot\Psi' + \betahat mc^2 -c \alphahat\cdot\eta)}\overline{\chi(\epsilon(\xi + \Psi'))}\chi(\epsilon(\xi + \Psi' - \eta))f(z)dw\dbar\eta
\end{align}
for $f \in \Sspace^N$ with $\Psi' = \Psi'(t,s;x,w+z,z)$, where $\eta \in R^d$ and $w \in R^d$.  Hence, taking account of \eqref{3.2}, we can prove Theorem 2.1 as in the proofs of Theorems 2.1 and 2.2 in \cite{Ichinose 2014}.
\par
	Next we will prove Theorem 2.2.  As proved in Lemma 6.1 of  \cite{Ichinose 1999}, under the assumptions of Theorem 2.A there exists an integer $M \geq 1$ such that we have \eqref{2.14}-\eqref{2.16}.  Therefore, we can see from Proposition 3.6 in \cite{Ichinose 2014} that under the assumptions of Theorem 2.2 we have \eqref{2.14}-\eqref{2.16}, and \eqref{3.2} or 
\begin{equation} \label{3.5.3}
            |\partial_x^{\alpha}\partial_y^{\beta}\partial_{z}^{\gamma}
            \Psi_j(t,s;x,y,z)|  \leq C_{\alpha,\beta,\gamma},
\ |\alpha + \beta  + \gamma| \geq 1
\end{equation}
in $\It^2\times \bR^{3d}$ for $j = 1,2,\dots,d$.
\par
Let  $\Gepsilon(t,s)$ be the operator defined  by \eqref{3.5}.
 The following proposition has already been shown in the proof of Proposition 3.2 of \cite{Ichinose 2014}.
\begin{pro}
 Assume \eqref{2.16} and 
\begin{equation*} 
|\partial_x^{\alpha}A_j(t,x)| \leq C_{\alpha}<x>^M
     \quad |\alpha| \geq 1
\end{equation*}
in $\It \times \bR^d$ for $j= 1,2,\dots,d$.  Then, $\bigl\{G_{\epsilon}(t,s)\bigr\}_{0 <\epsilon \leq 1}$ is a bounded family of operators from $\Sspace^N$ into itself and  there exists an operator $G(t,s)$ on $\Sspace^N$ independent of the choice of $\chi$ such that we have
\begin{equation}  \label{3.6}
G(t,s)f = \lim_{\epsilon\rightarrow 0}\Gepsilon(t,s)f
\end{equation}
 in $\Sspace^N$ for all $f \in \Sspace^N$ uniformly with respect to  $t$ and $s$ in $\It$. In particular, we have $G(s,s)f = f$ for all $f \in \Sspace^N$.
\end{pro}
   The following proposition has been stated as Theorem 5.2 of \cite{Ichinose 2014}, that had been proved in \cite{Ichinose 1995}.
   \begin{pro}
   Under the assumptions of Proposition 3.2 consider the Dirac equation \eqref{1.1} with $u(s) = f  \in B^a_{M+1}\ (a = 0,1,2,\dots)$ for $s  \in \It$.
   Then there exists a unique solution $U(t,s)f \in {\cal E}^0_t(\It;B^a_{M+1}) \cap {\cal E}^1_t(\It;
   B^{a-1}_{M+1})$, which
   satisfies
\begin{equation} \label{3.7}
    \Vert U(t,s)f \Vert = \Vert f \Vert, \  \Vert U(t,s)f \Vert_{B^a_{M+1}} \leq C_a(T) \Vert f \Vert_{B^a_{M+1}} \   (a = 1,2,\dotsc) 
  \end{equation}
  for $t$ and $s$ in $\It$.
   \end{pro}
   We have proved \eqref{3.2} or \eqref{3.5.3} under the assumptions of Theorem 2.2.  Hence we can prove the following as in the proof of Theorem 3.3 in \cite{Ichinose 2014}. 
	\begin{pro}
	Under the assumptions of Theorem 2.2   we have:
	  (1)  $G(t,s)$ defined in Proposition 3.2 can be extended to a bounded operator on $(L^2)^N$. (2)  There exists a constant $K_0 \geq 0$ such that 
\begin{equation}  \label{3.8}
\Vert G(t,s)f\Vert \leq e^{K_0(t-s)^2}\Vert f \Vert
\end{equation}
 for all $f \in L^2$ and $t, s \in \It$ with $|t -s| \leq 1$.
	\end{pro}
	
	\begin{rem}
   The inequality \eqref{3.8} above has been yielded directly from (3.16) in the proof of Theorem 3.3 of \cite{Ichinose 2014}.  As in the completely same way, we can prove
\begin{align*}
& \Vert G(t,s)f\Vert^2 = (f,f) + (t-s)^2(P(t,s;X,D_X,X')f,f) \\
& \geq \Vert f \Vert^2 - 2K_0(t-s)^2\Vert f \Vert^2 
\end{align*}
for $t$ and $s$ in $\It$ with $|t -s| \leq 1$.  This shows 
\begin{equation}  \label{3.9}
\Vert G(t,s)f\Vert^2 \geq e^{-4K_0(t-s)^2}\Vert f \Vert^2
\end{equation}
for $t$ and $s$ in $\It$ with $|t -s| \leq 1$ and  $4K_0(t - s)^2 \leq \log 2$, because   $1 - \theta \geq e^{-2\theta}$ holds for $0 \leq \theta \leq \log 2/2$.
	\end{rem}
  Noting Lemma 3.1,  we can prove the following as well as Proposition 3.4, as in the proof of Proposition 5.3 of \cite{Ichinose 2014} .
	\begin{pro}
	Under the assumptions of Theorem 2.2 we have
   \begin{equation} \label{3.10}
   \Vert G(t,s)f - U(t,s)f\Vert_{B^a_{M+1}} \leq C_a(t-s)^2 \Vert f\Vert_{B^{a+2}_{M+1}},\  -T \leq s, t \leq T
   \end{equation}
   for $a = 0,1,2,\dotsc$ and $f \in \Sspace^N$.
   \end{pro}
  Now, let us prove Theorem 2.2.  We take an electromagnetic potential $(V,A)$ satisfying \eqref{2.14}-\eqref{2.16} that induces $E(t,x)$ and $(B_{jk}(t,x))_{1 \leq j < k \leq d}$, as stated in the early part of this section.  For this $(V,A)$ we will prove the assertions (1) and (2) in Theorem 2.2.  The general case can be proved from these results by the use of the gauge transformation \eqref{2.12}, as in the proof of Theorem 2.1 of \cite{Ichinose 2014} on pp.506-507. \par
   For  $f \in \Sspace^N$ we can write \eqref{2.7} by using \eqref{3.5} and Proposition 3.2 as
\begin{align*}  
     K_{D\Delta}(t_f,t_i)f  & = \lim_{\epsilon \rightarrow 0}\Gepsilon(t_f,\tau_{\nu-1})\chi(\epsilon\cdot)\Gepsilon(\tau_{\nu-1},\tau_{\nu-2})\chi(\epsilon\cdot)\cdots \chi(\epsilon\cdot)
\Gepsilon(\tau_{1},t_i)f \\
& = G(t_f,\tau_{\nu-1})G(\tau_{\nu-1},\tau_{\nu-2})\cdots 
G(\tau_{1},t_i)f
\end{align*}
in $\Sspace^N$,  which proves
\begin{equation}   \label{3.11}
     K_{D\Delta}(t_f,t_i)f  = G(t_f,\tau_{\nu-1})G(\tau_{\nu-1},\tau_{\nu-2})\cdots 
G(\tau_{1},t_i)f
\end{equation}
in  $(L^2)^N$ for  $f \in (L^2)^N$ from Proposition 3.4.  Hence, applying  \eqref{3.8} to \eqref{3.11}, we can easily prove \eqref{2.19}  in Theorem 2.2 from \eqref{2.18}. 
Consequently, (1) in Theorem 2.2 has been proved.
\par
  Let $f \in (B^2_{M+1})^N$.  From \eqref{3.11} we can write
\begin{align}  \label{3.12}
   & K_{D\Delta}(t_f,t_i)f - U(t_f,t_i)f =  G(t_f,\tau_{\nu-1})\cdots G(\tau_1,t_i)f -   U(t_f,\tau_{\nu-1})\cdots U(\tau_1,t_i)f
       \notag \\
          & = \sum_{j=0}^{\nu-1}G(t_f,\tau_{\nu-1})\cdots G(\tau_{j+2},\tau_{j+1}) \bigl\{G(\tau_{j+1},\tau_{j}) -U(\tau_{j+1},\tau_{j})\bigr\}U(\tau_{j},t_i)f .
\end{align}
Let $\sigma(\Delta) \leq 1$ and apply Propositions 3.3-3.5 to the last equation in \eqref{3.12}.  Then we have
\begin{align}  \label{3.13}
   & \Vert K_{D\Delta}(t_f,t_i)f - U(t_f,t_i)f \Vert      \notag \\
& \leq \sum_{j=0}^{\nu-1}e^{K_0\sigma(\Delta)}C_0(\tau_{j+1}-\tau_{j})^2\Vert U(\tau_{j},t_i)f\Vert_{B^2_{M+1}} \notag\\
          & \leq C_0'\sigma(\Delta)e^{K_0\sigma(\Delta)}\Vert f\Vert_{B^2_{M+1}}.
\end{align}
 \par
   Let $f \in (L^2)^N$ and  $\sigma(\Delta) \leq 1$.  For an arbitrary constant $\epsilon > 0$   take a function $g \in (B^2_{M+1})^N$ such that $\Vert g - f\Vert < \epsilon$.   Then,
   using \eqref{2.19}, \eqref{3.7} and \eqref{3.13}, we can prove
\begin{align}  \label{3.14}
   & \Vert K_{D\Delta}(t_f,t_i)f - U(t_f,t_i)f \Vert  \leq \Vert K_{D\Delta}(t_f,t_i)g - U(t_f,t_i)g\Vert +  \Vert K_{D\Delta}(t_f,t_i)(f-g)\Vert
    \notag \\
          &  + \Vert  U(t_f,t_i)(f-g)\Vert \leq  C_0'\sigma(\Delta)e^{K_0\sigma(\Delta)}\Vert g\Vert_{B^2_{M+1}}+ e^{K_0\sigma(\Delta)}\Vert g - f\Vert + \Vert g-f\Vert ,
         \end{align}
         which shows 
\begin{equation*}  
\overline{  \lim_{\sigma(\Delta)\rightarrow 0}} \  \Vert K_{D\Delta}(t_f,t_i)f - U(t_f,t_i)f \Vert     \leq 2\epsilon.
\end{equation*}
Consequently we have been able to prove (2) of Theorem 2.2.  Therefore, the proof of Theorem 2.2 has been completed.
\section{Unitarity and Causality}
   In this section we will study the properties of the Feynman path integral $K_D(t_f,t_i)$  determined in Theorems 2.1 and 2.2.  
   \par
   First we will prove  the unitarity of $K_D(t_f,t_i)$ on $(L^2)^N$.   This result gives another proof of the unitarity of the fundamental solution $U(t_f,t_i)$ to \eqref{1.1}  on $(L^2)^N$ because of $U(t_f,t_i) = K_D(t_f,t_i)$ in Theorem 2.2, which is well known in the theory of partial differential equations.  We note that we can prove Theorem 2.2  without the use of the unitarity of $U(t_f,t_i)$.
   \begin{thm}
   Under the assumptions of Theorem 2.2  $K_D(t_f,t_i)$ is unitary on $(L^2)^N$.
   \end{thm}
   \begin{proof}
   We have proved \eqref{2.19}  in Theorem 2.2.  In the same way we can prove
\begin{equation*}  
    \Vert K_{D\Delta}(t_f,t_i)f \Vert \geq e^{-2K_0\sigma(\Delta)}\Vert f \Vert
    \end{equation*}
    for small $\sigma(\Delta)$ from \eqref{3.9} and \eqref{3.11}, which shows 
\begin{equation}  \label{4.1}
   e^{-2K_0\sigma(\Delta)}\Vert f \Vert \leq  \Vert K_{D\Delta}(t_f,t_i)f \Vert \leq e^{K_0\sigma(\Delta)}
   \Vert f \Vert.
    \end{equation}
Letting $\sigma(\Delta)$ tend to zero,  we obtain
\begin{equation}  \label{4.2}
    \Vert K_{D}(t_f,t_i)f \Vert = \Vert f \Vert
    \end{equation}
    for $f \in (L^2)^N$ from (2) of Theorem 2.2.  \par
    From \eqref{3.5}
    we can easily have
\begin{equation}  \label{4.3}
     (G_{\epsilon}(t,s)^*f)(x) = \iint e^{-iS(t,s;y,\xi,x)}f(y)
    \overline{ \chi(\epsilon \xi)}dy\dbar\xi
\end{equation}
for $f \in \Sspace^N$.   From \eqref{2.4} and \eqref{2.5} we can write 
\begin{align}  \label{4.4}
 &  S(t,s;x,\xi,y)           \notag\\
 &  = (x - y)\cdot\xi +  \int_{\bq^{t,s}_{x,y}}(A\cdot dx -Vdt)  - (t - s)(c\alphahat\cdot\xi + \betahat mc^2)
 \end{align} 
 as in the proof of (2.3) in \cite{Ichinose 1997},
 where
\begin{equation*} 
\bq^{t,s}_{x,y}: \bq^{t,s}_{x,y}(\theta) = (\theta,\qts(\theta)) \in \domain \ (s \leq \theta \leq t \ \text{or}\  t \leq \theta \leq s).
\end{equation*}
 This gives 
\begin{align*}  
 &  - S(t,s;y,\xi,x)  = -(y - x)\cdot\xi -  \int_{\bq^{t,s}_{y,x}}(A\cdot dx -Vdt)           \notag\\
 &  + (t - s)(c\alphahat\cdot\xi + \betahat mc^2) =  (x-y)\cdot\xi + \int_{\bq^{s,t}_{x,y}}(A\cdot dx -Vdt) \notag\\
 & -(s-t)(c\alphahat\cdot\xi + \betahat mc^2) = S(s,t;x,\xi,y),
 \end{align*} 
 which shows
\begin{equation}  \label{4.5}
     (G_{\epsilon}(t,s)^*f)(x) =  (G_{\epsilon}(s,t)f)(x)
\end{equation}
together with \eqref{4.3}.
Consequently the expression \eqref{3.11} indicates
\begin{equation}  \label{4.6}
     K_{D\Delta}(t_f,t_i)^*f= K_{D\Delta^*}(t_i,t_f)f
\end{equation}
for $f \in (L^2)^N$ with the  time-division $\Delta^*$ corresponding to $\Delta$, which proves
\begin{equation}  \label{4.7}
   \Vert K_{D}(t_f,t_i)^*f \Vert = \Vert K_{D}(t_i,t_f)f \Vert = \Vert f \Vert
    \end{equation}
    from (2) of Theorem 2.2 and \eqref{4.2}.   \par
    The equalities \eqref{4.2} and \eqref{4.7} imply that $K_{D}(t_f,t_i)$ is  unitary  on $(L^2)^N$, as well known.  In fact, it is easily seen from  the polarization identity (cf. p.63 of \cite{Reed-Simon}) that if and only if $F := K_{D}(t_f,t_i)$ is isometric  on $(L^2)^N$,
    $(Ff,Fg) = (f,g)$ are true for all $f$ and $g$ in  $(L^2)^N$, which is equivalent to $F^*F = \text{Identity}$  on $(L^2)^N$.  Since $F^*$ is also isometric, $FF^* = \text{Identity}$  on $(L^2)^N$ is yielded.  Thus, it has been proved that $F = K_{D}(t_f,t_i)$ is unitary on $(L^2)^N$.
   \end{proof}
   Secondly, we will prove that the Feynman path integral $K_D(t_f,t_i)f$ satisfies the causality principles, i.e. has the  speed not exceeding the velocity of light of propagation of disturbances.  This result gives another proof that  every solution to the Dirac equation \eqref{1.1} has the same property, which is also well known  in the theory of partial differential equations.
   For example, see  the 5th problem in \S 5.3 on p.170 of \cite{John}, Theorem 6.10 and its Note 2 on pp.364-365 of \cite{Mizohata} and \S 4 in Chapter IV on p.79 of \cite{Taylor}.  In all of these references, the method of proving the causality  principle is based on  the energy inequality and the introduction of a hypersurface spacelike with respect to the operator defining the equation. Thereby, a  delicate analysis is needed. \par
  Let $\alphahat^{(j)}\ (j = 1,2,\dots,d)$ be the $N\times N$ Hermitian matrix in \eqref{1.1} and $\lambda_k(\xi)\ (k =  1,2,\dots,N)$ the eigenvalue of the matrix $\alphahat\cdot\xi$, which is continuous  on  $\bR^d_{\xi}$.  We set
    \begin{equation} \label{4.8}
  \lambda_{\max} := \max_{j= 1,2,\dotsc,N} \max_{|\xi|=1}\lambda_j(\xi),
\end{equation} 
which is non-negative because of 
    \begin{equation} \label{4.9}
 \lambda_j(s\xi) = s \lambda_j(\xi)\quad(s \in \bR).
\end{equation} 
    \par
    For $f = {}^t(f_1,\dots,f_N) \in (L^2)^N$ we call the union $\cup_{j=1}^N \text{supp} f_j$ of the support of $f_j$ the support of $f$, and write it $\text{supp} f$.  For  a point $a \in \bR^d$ and $R \geq 0$ we write $\{x \in \bR^d;|x - a| \leq R\}$ as $B(a;R)$.  We have the following.
    \begin{thm}
    Let $K_D(t_f,t_i)f$ for $f \in (L^2)^N$ be the Feynman path integral determined in Theorem 2.2.  Then, $K_D(t_f,t_i)f$  has the speed not exceeding   $c \lambda_{\max}$ of propagation of disturbances.  That is,  if $\supp f$ is in $B(a;R)$, then  $\text{supp}\ K_D(t_f,t_i)f$ is  in  $B(a;c \lambda_{\max}|t_f -t_i|+R)$.
     \end{thm}
    The corollary below assures us that the Feynman path integral for the genuine Dirac equation satisfies the causality principle.
    \begin{cor}
    Besides the assumptions of Theorem 2.2 we suppose \eqref{1.2}.  Then the Feynman path integral $K_D(t_f,t_i)f$  for $f \in (L^2)^N$ has the speed not exceeding  $c$, the velocity  of light, of propagation of disturbances.  
    \end{cor}
    \begin{proof}
    From \eqref{1.2} we can easily have
    \begin{equation} \label{4.10}
    (\alphahat\cdot\xi)^2 = |\xi|^2\quad (\xi \in \bR^d)
    \end{equation}
    by the same argument as in \S 67 of \cite{Dirac}, which shows $\lambda_j(\xi)^2 = |\xi|^2$ and so
    $|\lambda_j(\xi)|= |\xi|$.  It follows from the hermiticity of  $\alphahat\cdot\xi$ that $\lambda_j(\xi)$ is real, which implies $\lambda_{\max} = 1$ from \eqref{4.8} and \eqref{4.9}.  Consequently we obtain Corollary 4.3 from Theorem 4.2.
    \end{proof}
    Now, we will state the well-known results as the Paley-Wiener theorem (cf. Theorem IX.11 on p.333 in \cite{Reed-Simon}) and Lie product formula (cf. Theorem VIII.29 on p.295 in \cite{Reed-Simon}) that will be used to prove Theorem 4.2. \par
\vspace{0.2cm}
{\bf Proposition 4.A} (Paley-Wiener). {\it Let $\zeta := \xi + i\eta \in \mathbb{C}^d$ be complex variables where $\xi \in \bR^d$ and $\eta \in \bR^d$.
An entire analytic function  $g(\zeta)$  on $\mathbb{C}^d$ is the Fourier transform $\widehat{f}(\zeta)  := \int e^{-i\zeta\cdot x}f(x)dx$ of a function  $f(x) \in C_0^{\infty}(\bR^d)$ with support in $B(0;R)$, if and only if for each $n = 1,2,\dots$ there is a constant $C_n \geq 0$ so that
\begin{equation*}
|g(\zeta)| \leq \frac{C_ne^{R|\eta|}}{(1 + |\zeta|)^n}
\end{equation*}
for all $\zeta \in \mathbb{C}^d$. }\par
\vspace{0.2cm}
    Let $A$ be an $N\times N$ matrix.  We write its norm $\sup_{|u|=1}|Au|$ as $\Vert A \Vert$, where $u = {}^t(u_1,\dots,u_N) \in \mathbb{C}^N$ and $|u| = \sqrt{\sum_{j=1}^N|u_j|^2}$. \par
\vspace{0.2cm}
{\bf Proposition 4.B} (Lie product formula). {\it Let $A$ and $B$ be finite dimensional matrices.   Then we have
\begin{equation*}
\exp (A + B) = \lim_{n\rightarrow \infty}\left[\exp \frac{A}{n} \exp \frac{B}{n}\right]^n
\end{equation*}
in the topology of the norm.
}\par
\vspace{0.2cm}
   The following lemma is  essential for the proof of Theorem 4.2.
   \begin{lem}
   Let $\lambda_{\max}$ be the constant defined by \eqref{4.8}.  Let $\rho \in \bR$ and $u = {}^t(u_1,\dots,u_N)\in \mathbb{C}^N$.  Then we have
\begin{equation*}
\left| \exp \Bigl(-i\rho c\alphahat\cdot(\xi + i\eta)\Bigr) u\right|  \leq \Bigl(\exp  (|\rho |c| \eta|\lambdamax)\Bigr)|u|.
\end{equation*}
   \end{lem}
   \begin{proof}
   Let $\eta = 0$.  Then 
\begin{equation*}
\left| \exp \Bigl(-i\rho c\alphahat\cdot(\xi + i\eta)\Bigr) u\right|  = \left| \exp (-i\rho c\alphahat\cdot \xi)u\right| = |u|
\end{equation*}
holds since $\exp (-i\rho c\alphahat\cdot \xi)$ is unitary, which shows Lemma 4.4. \par
   Let $\eta \not= 0$.  Since $\alphahat\cdot\eta$ is Hermitian, we can have a diagonal matrix 
   \begin{equation*}
   \mathfrak{U}^{-1}(\alphahat\cdot\eta)\mathfrak{U} = 
   \begin{pmatrix}
   \lambda_1(\eta) & 0  & 0 & \cdots & 0\\
    0 & \lambda_2(\eta) &  0 & \cdots & 0 \\
    \hdotsfor{4} & 0\\
    0 & 0 &  0 & \cdots & \lambda_N(\eta)
    \end{pmatrix}
   \end{equation*}
   by using a unitary matrix $\mathfrak{U}$.  Consequently we get
   \begin{align*}
   & \mathfrak{U}^{-1}\exp(\rho c \alphahat\cdot\eta)\mathfrak{U}  = \exp \rho c
   \begin{pmatrix}
   \lambda_1(\eta) & 0  & 0 & \cdots & 0\\
    0 & \lambda_2(\eta) &  0 & \cdots & 0 \\
    \hdotsfor{4} & 0\\
    0 & 0 &  0 & \cdots & \lambda_N(\eta)
    \end{pmatrix}\\
    & = \exp \rho c |\eta|
   \begin{pmatrix}
   \lambda_1(\eta/|\eta|) & 0  & 0 & \cdots & 0\\
    0 & \lambda_2(\eta/|\eta|) &  0 & \cdots & 0 \\
    \hdotsfor{4} & 0\\
    0 & 0 &  0 & \cdots & \lambda_N(\eta/|\eta|)
    \end{pmatrix}\\
    & =
   \begin{pmatrix}
    e^{\rho c |\eta|\lambda_1(\eta/|\eta|)}  & 0  & 0 & \cdots & 0\\
    0 & e^{\rho c |\eta|\lambda_2(\eta/|\eta|)} &  0 & \cdots & 0 \\
    \hdotsfor{4} & 0\\
    0 & 0 &  0 & \cdots & e^{\rho c |\eta|\lambda_N(\eta/|\eta|)}    \end{pmatrix}
   \end{align*}
together with \eqref{4.9},   which yields
   \begin{align*}
  & |\exp(\rho c \alphahat\cdot\eta)u| = |\mathfrak{U}\mathfrak{U}^{-1}\exp(\rho c \alphahat\cdot\eta)\mathfrak{U}
  \mathfrak{U}^{-1}u| \\
  & =  |\mathfrak{U}^{-1}\exp(\rho c \alphahat\cdot\eta)\mathfrak{U}\mathfrak{U}^{-1}u| \leq \Bigl( \exp (|\rho|c|\eta|\lambdamax)\Bigr)|\mathfrak{U}^{-1}u| \\
  & =\Bigl( \exp (|\rho|c|\eta|\lambdamax)\Bigr)|u|
   \end{align*}
   by using the unitarity of $\mathfrak{U}$.  Hence we obtain
   \begin{equation} \label{4.11}
   |\exp(\rho c \alphahat\cdot\eta)u| \leq \Bigl( \exp (|\rho|c|\eta|\lambdamax)\Bigr)|u|.
   \end{equation}
   \par
   Now, Proposition 4.B indicates
   \begin{align} \label{4.12}
 & \exp \Bigl(-i\rho c \alphahat\cdot(\xi + i\eta)\Bigr) u = \exp (\rho c \alphahat\cdot \eta - i\rho c \alphahat\cdot \xi) u \notag\\
  & = \lim_{n\rightarrow \infty}\left[\exp \frac{\rho c \alphahat\cdot \eta}{n}\exp \frac{-i\rho c \alphahat\cdot \xi}{n}\right]^nu \quad \text{in}\ \mathbb{C}^N.
   \end{align}
   Noting \eqref{4.11} and the unitarity of $\exp (-i\rho c \alphahat\cdot \xi/n)$, we can easily prove
   \begin{align*} 
 & |\exp (\rho c \alphahat\cdot \eta/n) \exp (-i\rho c \alphahat\cdot \xi/n)u| \\
 & \leq \Bigl(\exp \bigl(|\rho| c |\eta|\lambdamax/n\bigr)\Bigr)|\exp (-i\rho c \alphahat\cdot \xi/n)u| = \Bigl(\exp \bigl(|\rho| c |\eta|\lambdamax/n\bigr)\Bigr)|u|.
   \end{align*}
In the same way we have
   \begin{align*} 
 & \left|\Bigl[\exp (\rho c \alphahat\cdot \eta/n) \exp (-i\rho c \alphahat\cdot \xi/n)\Bigr]^2u\right| \\
 & \leq \Bigl(\exp \bigl(|\rho| c |\eta|\lambdamax/n\bigr)\Bigr) |\exp (\rho c \alphahat\cdot \eta/n) \exp (-i\rho c \alphahat\cdot \xi/n)u|   \\
 & \leq \Bigl(\exp \bigl(2|\rho| c |\eta|\lambdamax/n\bigr)\Bigr)|u|.
  \end{align*}
Repeating this argument, we can prove
   \begin{align} \label{4.13}
 & \left|\Bigl[\exp (\rho c \alphahat\cdot \eta/n) \exp (-i\rho c \alphahat\cdot \xi/n)\Bigr]^nu\right| \notag \\
 & \leq \Bigl(\exp \bigl(|\rho| c |\eta|\lambdamax\bigr)\Bigr)|u|,
  \end{align}
   which completes the proof of  Lemma 4.4 together with \eqref{4.12}.
   \end{proof}
   \begin{lem}  Let $\rho \in \bR$ and $u = {}^t(u_1,\dots,u_N) \in \mathbb{C}^N$.  Then we have
\begin{equation*}
\left| \exp \Bigl(-i\rho \bigl\{ c\alphahat\cdot(\xi + i\eta) + \betahat mc^2\bigr\}\Bigr) u\right|  \leq\Bigl(\exp \bigl(|\rho| c |\eta|\lambdamax\bigr)\Bigr)|u|.
\end{equation*}
\end{lem}
\begin{proof}
We have
   \begin{align} \label{4.14}
 & \exp \Bigl(-i\rho \bigl\{ c\alphahat\cdot(\xi + i\eta) + \betahat mc^2\bigr\}\Bigr) u  \notag\\
  & = \lim_{n\rightarrow \infty}\left[\exp \frac{-i\rho c \alphahat\cdot (\xi + i\eta)}{n}\exp \frac{-i\rho \betahat mc^2}{n}\right]^nu \quad \text{in}\ \mathbb{C}^N
   \end{align}
   from Proposition 4.B.  Lemma 4.4 shows
   \begin{align*} 
 & \left|\exp \Bigl(-i\rho c \alphahat\cdot (\xi + i\eta)/n\Bigr) \exp \Bigl(-i\rho  \betahat mc^2/n\Bigr)u\right| \\
 & \leq \Bigl(\exp \bigl(|\rho| c |\eta|\lambdamax/n\bigr)\Bigr)\left| \exp \Bigl(-i\rho  \betahat mc^2/n\Bigr)u\right| = \Bigl(\exp \bigl(|\rho| c |\eta|\lambdamax/n\bigr)\Bigr)|u|
   \end{align*}
   because of the unitarity of $\exp (-i\rho  \betahat mc^2/n)$.  Hence we can prove
   \begin{align} \label{4.15}
 & \left|\left[\exp \Bigl(-i\rho c \alphahat\cdot (\xi + i\eta)/n\Bigr) \exp \Bigl(-i\rho  \betahat mc^2/n\Bigr)\right]^nu\right| \notag\\
 & \leq  \Bigl(\exp \bigl(|\rho| c |\eta|\lambdamax\bigr)\Bigr)|u|
   \end{align}
as in the proof of \eqref{4.13}, which completes the proof of Lemma 4.5 together with \eqref{4.14}.
\end{proof}
   Taking a function $\psi(x) \in C_0^{\infty}(\bR^d)$ with support in $B(0;1)$ and $\displaystyle{\int} \psi(x)dx = 1$,  we define $\chi(\xi) \in \Sspace$ by its Fourier transform $\widehat{\psi}(\xi)$.  Then $\chi(0) = 1$ holds.  We fix this $\chi(\xi)$ hereafter.  For $\epsilon > 0$ and $f \in \Sspace^N$ let us write
\begin{align}  \label{4.16}
    &  \left(G_{\epsilon}^0(t,s)f\right) (x) := \iint e^{i(x-y)\cdot\xi-i\rho(c\alphahat\cdot \xi  + \betahat mc^2)}  f(y)\chi(\epsilon\xi)dy\dbar\xi     \notag \\
       &  = \int  e^{ix\cdot\xi-i\rho(c\alphahat\cdot \xi + \betahat mc^2)}\widehat{f}(\xi)\chi(\epsilon\xi) \dbar\xi, 
       \  \ \rho = t - s,
\end{align}
which is equal to $\Gepsilon(t,s)f$ defined by \eqref{3.5} with $V = 0$ and $A = 0$.
\begin{pro}
Let $f \in C_0^{\infty}(\bR^d)^N$ with support in $B(0;R)$.  Then we have $\supp\  G_{\epsilon}^0(t,s)f \subset B(0;c\lambdamax|t - s| + R + \epsilon).$
\end{pro}
\begin{proof}
The expression \eqref{4.16} gives that the Fourier transform of $G_{\epsilon}^0(t,s)f$ is 
\begin{equation}  \label{4.17}
      v_{\epsilon}(t,s;\xi) := e^{-i\rho(c\alphahat\cdot \xi + \betahat mc^2)}\widehat{f}(\xi)\chi(\epsilon\xi).
    \end{equation}
    Proposition 4.A indicates that $\widehat{f_j}(\xi)\ (j = 1,2,\dots,N)$ and $\chi(\xi)$ can be extended to entire functions on $\mathbb{C}^d$ and satisfy
\begin{equation} \label{4.18}
 |\widehat{f_j}(\zeta)| \leq \frac{C_ne^{R|\eta|}}{(1 + |\zeta|)^n}, \quad |\chi(\zeta)| \leq \frac{C_ne^{|\eta|}}{(1 + |\zeta|)^n}
\end{equation}
for each $n = 1,2,\dots$ with a constant $C_n \geq 0$.  Hence, applying Lemma 4.5 to \eqref{4.17}, we can see by \eqref{4.18} that $v_{\epsilon}(t,s;\xi)$ can be extended to an entire function on $\mathbb{C}^d$ and satisfies 
\begin{equation} \label{4.19}
 |v_{\epsilon}(t,s;\zeta)| \leq e^{|\rho|c|\eta|\lambdamax}|\widehat{f}(\zeta)||\chi(\epsilon\zeta)|
 \leq  \frac{C_n^2e^{(|\rho|c\lambdamax + R + \epsilon)|\eta|}}{(1 + |\zeta|)^{n}(1 + |\epsilon\zeta|)^{n}}
\end{equation}
for each $n$, which proves Proposition 4.6 from Proposition 4.A.
\end{proof}
\begin{cor}
Let $f \in \Czerospace^N$ with support in $B(a;R)$ for $a \in \bR^d$.  Then we have $\supp\  G_{\epsilon}^0(t,s)f \subset B(a;c\lambdamax|t - s| + R + \epsilon)$.
\end{cor}
\begin{proof}
Set $g(x) := f(x + a)$.  
Then $\widehat{g}(\xi) = e^{ia\cdot\xi}\widehat{f}(\xi)$ and  $\supp\ g \subset B(0;R)$.  Consequently from \eqref{4.16} we can see
\begin{align*}  
    &  \left(G_{\epsilon}^0(t,s)f\right) (x) = \int  e^{ix\cdot\xi-i\rho(c\alphahat\cdot \xi + \betahat mc^2)}
    e^{-ia\cdot\xi}\widehat{g}(\xi)\chi(\epsilon\xi) \dbar\xi    \notag \\
       &  =  \left(G_{\epsilon}^0(t,s)g\right) (x-a). 
\end{align*}
Hence Corollary 4.7 is yielded
since $\supp\  G_{\epsilon}^0(t,s)g \subset B(0;c\lambdamax|t - s| + R + \epsilon)$ follows from Proposition 4.6.
\end{proof}
\begin{pro}
Let $t$ and $s$ be in $\It$. 
We consider the operator $G(t,s)$ on $\Sspace^N$ defined in Proposition 3.2.  Then $G(t,s)f$ for $f \in \Czerospace^N$ has the speed not exceeding $c\lambdamax$ of propagation of disturbances.
\end{pro}
\begin{proof}
 Let $f \in \Czerospace^N$ with support in $B(a;R)$. Take $\varphi_R(x) \in C_0^{\infty}(\bR^d)$ such that $\varphi_R(x) = 1$ if $|x| \leq R$.  Noting \eqref{2.4}, \eqref{2.5} and $\supp\ f \subset B(a;R)$, we can write $\Gepsilon(t,s)f$ defined by \eqref{3.5} as
\begin{equation}  \label{4.20}
     \left(G_{\epsilon}(t,s)f\right) (x) = \iint e^{i(x-y)\cdot\xi-i\rho(c\alphahat\cdot \xi  + \betahat mc^2)} w(t,s;x,y) f(y)\chi(\epsilon\xi)dy\dbar\xi,  
     \end{equation}
     where
\begin{align}  \label{4.21}
    & w(t,s;x,y) = \varphi_R(y-a)\exp \Bigl\{i(x-y)\cdot \int_0^1A(t-\theta\rho,x-\theta(x-y))d\theta \notag\\
     &  -i\rho\int_0^1V(t-\theta\rho,x-\theta(x-y))d\theta \Bigr\}.
     \end{align}
     Let $\alpha$ be multi-indices such that $|\alpha| = 2d$.  It follows from the assumptions of Proposition 3.2 that we have
\begin{align} \label{4.22}
& |\partial_y^{\alpha}w(t,s;x,y)| \leq C(<x>^{M+1}<x-y>^{M+1})^{2d}\sum_{|\beta|\leq 2d} 
|\partial_y^{\beta}\varphi_R(y-a)|    \notag\\
& \leq C'(<x>^{2(M+1)}<y>^{M+1})^{2d}\sum_{|\beta|\leq 2d} 
|\partial_y^{\beta}\varphi_R(y-a)|  \notag \\
& \leq C'<R + |a|>^{2(M+1)d}<x>^{4(M+1)d}\sum_{|\beta|\leq 2d} 
|\partial_y^{\beta}\varphi_R(y-a)|,
\end{align}
where we used $<x - y> \leq \sqrt{2}<x><y>$.  \par
   Using $\supp\ w(t,s;x,\cdot) \subset B(a;R)$, we can expand $w(t,s;x,y)$ into a Fourier series with respect to variables $y \in B(a;R)$ 
\begin{equation}  \label{4.23}
    w(t,s;x,y) = \sum_{n_1=-\infty}^{\infty} \cdots \sum_{n_d=-\infty}^{\infty}c_n(t,s;x)e^{in\omega\cdot(y-a)},
     \end{equation}
\begin{equation}  \label{4.24}
   c_n(t,s;x)=\left (\frac{1}{2l}\right)^d\int_I    w(t,s;x,y)e^{-in\omega\cdot(y-a)}dy,
   \end{equation}
     where $n = (n_1,\dots,n_d)$, $l \geq R$ is a constant, $\omega = \pi/l$ and $I$ is a cube in $\bR^d$ with edges of length $2l$.  From \eqref{4.22} and \eqref{4.24} we have
\begin{equation}  \label{4.25}
   |c_n(t,s;x)| \leq \frac{C<x>^{4(M+1)d}}{n_1^2\cdots n_d^2}.
   \end{equation}
Hence, using \eqref{4.20} and \eqref{4.23}, we can write $(\Gepsilon(t,s)f)(x)$ as an infinite series 
\begin{align}  \label{4.26}
     & \sum_{n_1=-\infty}^{\infty} \cdots \sum_{n_d=-\infty}^{\infty}c_n(t,s;x)\iint e^{i(x-y)\cdot\xi-i\rho(c\alphahat\cdot \xi  + \betahat mc^2)} e^{in\omega\cdot(y-a)} \notag \\
     &\quad \times  f(y)\chi(\epsilon\xi)dy\dbar\xi,
     \end{align}
     which converges uniformly on compact sets in $\bR_x^d$.  Consequently we see
\begin{equation}  \label{4.27}
  \supp\ \Gepsilon(t,s)f \subset B(a;c\lambdamax|t - s| + R + \epsilon)
    \end{equation}
     since
     Corollary 4.7 shows that the support of each term in \eqref{4.26} is in $B(a;\\ c\lambdamax|t - s| + R + \epsilon)$.  Therefore, we can complete the proof of Proposition 4.8 from Proposition 3.2.
\end{proof}
   Now, we will prove Theorem 4.2.  Let us use the gauge transformation as in the proof of Theorem 2.2.  Hence we may assume \eqref{2.14}-\eqref{2.16}. Suppose $t_i < t_f$.  Another case can be proved in the same way.  Take a 
   time-division $\{\tau_j\}_{j=1}^{\nu-1}$ satisfying 
\begin{equation}  \label{4.28}
  t_i <\tau_1 <\tau_2 < \dots < \tau_{\nu-1} < t_f.
    \end{equation}
     Then, applying Proposition 4.8 to each $G(\tau_j,\tau_{j-1})$ in \eqref{3.11}, we can see that $K_{D\Delta}(t_f,t_i)f$ for $f \in \Czerospace^N$ has the speed not exceeding $c\lambdamax$ of propagation of disturbances and so does $K_{D\Delta}(t_f,t_i)f$ for $f \in (L^2)^N$, which  can be proved from \eqref{2.19} by making $f$ approximated in $(L^2)^N$ by functions in $\Czerospace^N$.
     Therefore, we have been able to complete the proof of Theorem 4.2 from (2) in Theorem 2.2.
\end{document}